\documentclass[a4paper,english,cleveref,autoref,english]{lipics-v2019}

\usepackage[utf8]{inputenc}
\usepackage[T1]{fontenc}
\usepackage{amsmath}
\usepackage{amssymb}
\usepackage{xspace,cleveref}
\usepackage{stmaryrd}
\usepackage{comment}
\usepackage{mathtools}

\usepackage[textwidth=2cm,textsize=small]{todonotes}

\newcommand*{\eg}{e.g.\@\xspace}

\newcommand{\br}{\mathbf{r}}
\newcommand{\bu}{\mathbf{u}}
\newcommand{\bv}{\mathbf{v}}

\newcommand{\cA}{\mathcal{A}}
\newcommand{\Oh}{\mathcal{O}}

\newcommand{\bbN}{\mathbb{N}}
\newcommand{\N}{\bbN}
\newcommand{\bbZ}{\mathbb{Z}}
\newcommand{\bbQ}{\mathbb{Q}}
\newcommand{\bbD}{\mathbb{D}}
\newcommand{\bbS}{\mathbb{S}}
\newcommand{\bbF}{\mathbb{F}}
\newcommand{\bbK}{\mathbb{K}}
\newcommand{\bbL}{\mathbb{L}}

\newcommand{\fZ}{\mathfrak{Z}}
\newcommand{\fF}{\mathfrak{F}}

\renewcommand{\leq}{\leqslant}
\renewcommand{\geq}{\geqslant}
\renewcommand{\le}{\leqslant}
\renewcommand{\ge}{\geqslant}

\newcommand{\wh}[1]{\widehat{#1}}
\newcommand{\set}[1]{\{ #1 \}}
\newcommand{\sem}[1]{\left\llbracket #1 \right\rrbracket}
\newcommand{\transpose}{\mathsf{T}}
\newcommand{\td}[1]{\tilde{#1}}

\title{On polynomial recursive sequences}

\author{Micha\"{e}l Cadilhac}{DePaul University,
  Chicago, IL, USA}{michael@cadilhac.name}{https://orcid.org/0000-0001-9828-9129}{}
\author{Filip Mazowiecki}{Max Planck Institute for Software Systems, Germany}{filipm@mpi-sws.org}{}{}
\author{Charles Paperman}{Universit\'{e} de Lille, France}{charles.paperman@univ-lille.fr}{}{}
\author{Micha\l{} Pilipczuk}{University of Warsaw, Poland}{michal.pilipczuk@mimuw.edu.pl}{}%
       {This work is a part of project TOTAL that has received funding from the 
        European Research Council (ERC) under the European Union's Horizon 2020 
        research and innovation programme, grant agreement No.~677651.}
\author{G{\'{e}}raud S{\'{e}}nizergues}{Universit\'{e} de Bordeaux, France}{geraud.senizergues@u-bordeaux.fr}{}{}

\authorrunning{M. Cadilhac, F. Mazowiecki, Ch. Paperman, Mi. Pilipczuk and G. S{\'{e}}nizergues}
\Copyright{Micha\"{e}l Cadilhac, Filip Mazowiecki, Charles Paperman, Micha\l{} Pilipczuk and G{\'{e}}raud S{\'{e}}nizergues}
        
\category{}
\relatedversion{}
\supplement{}
\nolinenumbers

\EventEditors{???}
\EventNoEds{2}
\EventLongTitle{???}
\EventShortTitle{???}
\EventAcronym{???}
\EventYear{???}
\EventDate{???}
\EventLocation{????}
\EventLogo{}
\SeriesVolume{??}
\ArticleNo{??}

\begin{document}

\maketitle

\begin{abstract}
  We study the expressive power of {\em{polynomial recursive sequences}}, a
  nonlinear extension of the well-known class of linear recursive sequences.
  These sequences arise naturally in the study of nonlinear extensions of
  weighted automata, where (non)expressiveness results translate to class
  separations.  A typical example of a polynomial recursive sequence is
  $b_n=n!$.  Our main result is that the sequence $u_n=n^n$ is not polynomial
  recursive.
\end{abstract}

\begin{picture}(0,0)
\put(392,25)
{\hbox{\includegraphics[width=40px]{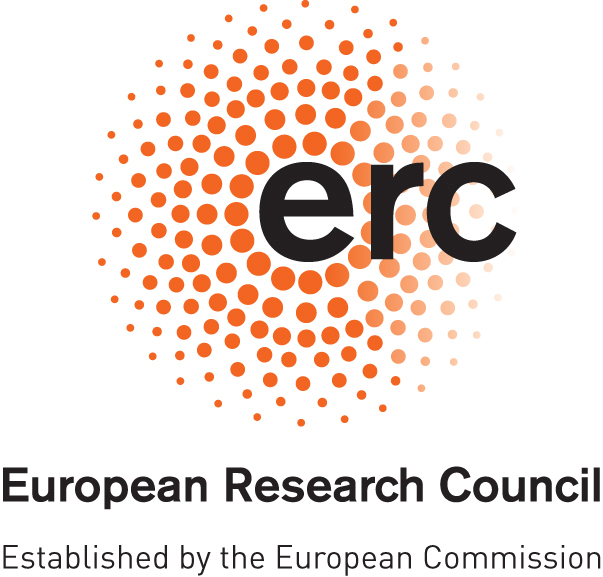}}}
\put(382,-35)
{\hbox{\includegraphics[width=60px]{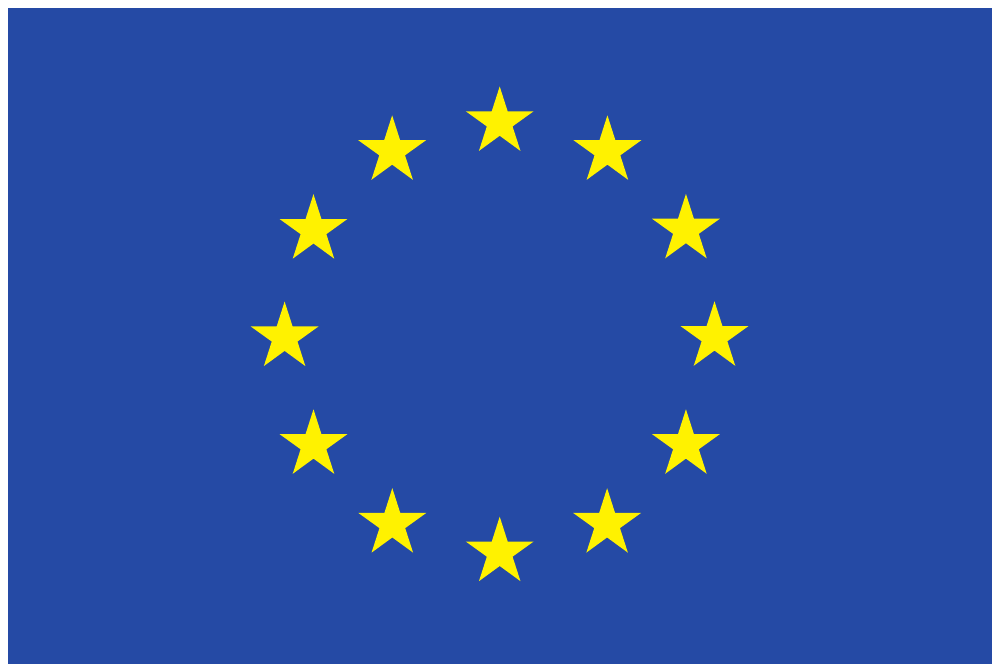}}}
\end{picture}


\tableofcontents

\newpage

\section{Introduction}
\label{sec:introduction}
Sequences defined recursively arise naturally in many areas, particularly in mathematics and computer science.
One of the most studied classes is that of \emph{linear recursive sequences}. 
Such sequences are defined by fixing the values of the first $k$ elements, while every subsequent element can be obtained as a linear combination of the $k$ elements preceding it.
The most famous example is the Fibonacci sequence, defined by setting $f_0 = 0$, $f_1 = 1$, and the recurrence relation $f_{n+2} = f_{n+1} + f_n$. 

It is well known that every linear recursive sequence can be defined by a system of $k$ mutually recursive sequences, 
where for every sequence we fix the initial value and provide a recurrence relation expressing the $(n+1)$st element as a linear combination of the $n$th elements of all the sequences~\cite{halava2005skolem}.
For example, to define the Fibonacci sequence $f_n$ in this way, one needs one
auxiliary sequence: we set $f_0 = 0$, $g_0 = 1$, and postulate
\begin{equation}\label{eq:fibo}
\begin{cases}
f_{n+1} = g_{n},  \\
g_{n+1} = f_{n} + g_{n}.
\end{cases}
\end{equation}

In this paper we study \emph{polynomial recursive sequences} over rational numbers that generalise linear recursive sequences. 
They are defined by systems of sequences like~\eqref{eq:fibo}, but on the right hand side we allow arbitrary polynomial expressions, rather than just linear combinations.
For example, the sequence $b_n = n!$ can be defined in this way using one auxiliary sequence: we may set $b_0 = c_0 = 1$ and write
\begin{equation}\label{eq:factorial}
\begin{cases}
b_{n+1} = b_{n} \cdot c_{n}, \\
c_{n+1} = c_{n} + 1.
\end{cases}
\end{equation}
Thus, the recurrence relation uses two polynomials: $P_1(x_1,x_2) = x_1x_2$ and $P_2(x_1,x_2) = x_2 + 1$.

The two classes of linear and polynomial recursive sequences appear naturally in automata theory, and in particular in connection with weighted automata and higher-order pushdown automata.
\emph{Weighted automata} over the rational semiring are a quantitative variant of finite automata that assign rational numbers to words~\cite{DrosteHWA09}. 
In the special case of a $1$-letter alphabet, each word can be identified with its length. 
Then a weighted automaton defines a mapping from natural numbers (possible
lengths) to rationals, and this can be seen as a sequence.
It is known that sequences definable in this way by weighted automata are exactly the linear recursive sequences~\cite{BarloyFLM20}.
\emph{Pushdown automata of order k} can be used for defining mappings from words
to words~\cite{Senizergues07}; in particular, for $k=2$ and 1-letter alphabets, such automata compute exactly the linear recursive sequences of natural integers~\cite{FerMarSen14}.

Thus, nonlinear extensions of linear recursive sequences may correspond to nonlinear extensions of weighted automata.
For the latter, consider three examples:
\begin{itemize}
 \item \emph{polynomial recurrent relations} that generalise pushdown automata of order~3~\cite{FraSen06,Senizergues07};
 \item \emph{cost-register automata} which arose as a variant of streaming transducers~\cite{AlurC10,AlurDDRY13};
 \item \emph{polynomial automata}, connected to reachability problems for vector addition systems~\cite{BenediktDSW17}.
\end{itemize}
Surprisingly, these three models, although introduced in different contexts,
are all equivalent.\footnote{This is a simple but technical observation as the three models are essentially syntactically equivalent. 
Throughout the paper we will use the name {\em{cost-register automata}} to refer to all three models.} 
Moreover, over unary alphabets
they define exactly polynomial recursive sequences, in the same fashion as
weighted automata (respectively order 2 pushdown automata) over unary alphabets define
linear recursive sequences.


The goal of this paper is to study the expressive power of polynomial recursive sequences.
Clearly, this expressive power extends that of linear recursive sequences: it is easy to see that every linear recursive sequence has growth bounded by $2^{\Oh(n)}$, while already the sequence $b_n=n!$ grows
faster. In fact, already the recurrence relation $a_0=2,\ a_{n+1}=\left(a_n\right)^2$ defines the sequence $2^{2^n}$, whose growth is doubly-exponential.
However, there are well-known integer sequences related to these examples for which definability as a polynomial recursive sequence seems much less clear.
The first example is the sequence $u_n=n^n$.
The second example is the sequence of Catalan numbers $C_n=\frac{1}{n+1}\binom{2n}{n}$.
Note that by Stirling's approximation, $n^n$ is asymptotically very close to $n!$, while $C_n$ is, up to factors polynomial in $n$, roughly equal to $4^n$.
For these reasons, simple asymptotic considerations cannot prove the sequences
$u_n=n^n$ and $C_n$ to be not polynomial recursive.
Recall that the Catalan numbers admit multiple combinatorial interpretations, which can be used to derive the recurrence formulas
$C_{n+1}=\sum_{i=0}^n C_iC_{n-i}$ and $(n+2)C_{n+1} = (4n+2)C_n$. Note that these formulas are {\em{not}} of the form
of recurrence formulas considered in this work.
Additionally, it is known that Catalan numbers $C_n$ are {\em{not}} linear recursive (see \eg~\cite{BhattiproluGV17}), despite having growth $2^{\Oh(n)}$.

\subparagraph*{Our results.}
We show that both the sequence of Catalan numbers $C_n$ and the sequence $u_n=n^n$ are not polynomial recursive.
For this, we present two techniques for proving that a sequence is {\em{not}} polynomial recursive.
The first technique for Catalan numbers is number-theoretical: we show that a polynomial recursive sequence of integers is ultimately periodic modulo any large enough prime.
The second technique for $n^n$ is more algebraic in nature: we show that for every polynomial recursive sequence there exists $k\in \N$ such that every $k$ consecutive elements of the sequence satisfy a nontrivial polynomial equation.
The fact that $u_n=n^n$ is not polynomial recursive is our main result.
These inexpressibility results were announced without proofs by the
fifth coauthor in an invited talk in 2007~\cite{Senizergues07}.  The present paper
contains proofs and extensions of these results.

\subparagraph*{Applications.}
The discussed models of cost-register automata~\cite{FraSen06,AlurDDRY13,BenediktDSW17} are  not the only nonlinear extensions of weighted automata that appear in the literature. 
We are aware of at least two more extensions: weighted context-free grammars~\cite{baker1979trainable,BhattiproluGV17} and weighted MSO logic~\cite{DrosteG07,KreutzerR13}. 
As it happens, over the $1$-letter alphabet,
weighted context-free grammars can define Catalan numbers, and weighted MSO logic can define $n^n$.
Therefore, as a corollary of our results we show that functions expressible in pushdown-automata of level 4,
weighted context-free grammars and weighted MSO logic are not always expressible in the class of cost-register automata.

The class of holonomic sequences is another extension of linear recursive sequences~\cite{KauersP11}. 
These sequences are defined recursively with one sequence, but the coefficients in the recursion are polynomials of the element's index. For example, $b_n = 1$ and $b_{n+1} = (n+1)b_n$ defines $b_n = n!$.
The expressiveness of this class has also been studied and in particular the sequence $n^n$ is known to be not in the class of holonomic sequences~\cite{Gerhold04}.
As a corollary of our results one can show that there are no inclusions between the classes of holonomic sequences and polynomial recursive sequences. 
On the one hand every holonomic sequence is asymptotically bounded by $2^{p(n)}$ for some polynomial $p$~\cite{KauersP11}, and the sequence $a_n = 2^{2^n}$ is polynomial recursive.
On the other hand, Catalan numbers admit a definition as a holonomic sequence: $C_0 = 1$ and $(n+2)C_{n+1} = (4n+2)C_n$.
In Section~\ref{sec:conclusion} we discuss the class of \emph{rational recurrence sequences} that generalises both holonomic and polynomial recursive sequences.



\subparagraph*{Organisation.}
In Section~\ref{sec:preliminaries} we give basic definitions and examples of linear and polynomial recursive sequences.
In Section~\ref{sec:simple} we show that the definition of polynomial recursive sequences requires a system of sequences and, unlike linear recursive sequences, cannot be equivalently defined using only one sequence. 
Then in Sections~\ref{sec:periodicity} and~\ref{sec:cancelling} we show that the sequence of Catalan numbers $C_n$ and the sequence $u_n=n^n$ are not polynomial recursive.
We conclude in Section~\ref{sec:conclusion}. In Appendix~\ref{sec:weighted} we explain in details our corollaries for weighted automata.

\section{Preliminaries}
\label{sec:preliminaries}
By $\bbN$ we denote the set of nonnegative integers.
A {\em{sequence}} over a set $\bbD$ is a function $u \colon \bbN \to \bbD$; all the sequences considered in this work are over the field of rationals $\bbQ$.
We use the notation $\langle u_n \rangle_{n \in \bbN}$ for elements of sequences, where $u_n=u(n)$.
Also, we use bold-face letters as a short-hand for sequences, e.g., $\bu = \langle u_n
\rangle_{n \in \bbN}$.

We now introduce the two main formalisms for describing sequences: linear recursive sequences and polynomial recursive sequences.


\subparagraph*{Linear recursive sequences.}
A {\em{$k$-variate linear form}} (or \emph{linear form} if $k$ is irrelevant) over $\bbQ$ is a function $L\colon \bbQ^k\to \bbQ$ of the form
\[L(x_1,\ldots,x_k)=a_1x_1+\ldots+a_kx_k\]
for some $a_1,\ldots,a_k\in \bbQ$.
A sequence of rationals $\bu$ is a {\em{linear recursive sequence}} if there exist $k\in \N$ and a $k$-variate linear form $L$ such that $\bu$ satisfies the recurrence relation
\begin{align}\label{def:lrs1}
u_{n+k} = L(u_{n},\ldots,u_{n+k-1})\qquad\textrm{for all }n\in \bbN.
\end{align}
Observe that such a sequence is uniquely determined by the form $L$ and its first $k$ elements: $u_0,\ldots,u_{k-1} \in \bbQ$. 
The minimal $k$ for which a description of $\bu$ as in~\eqref{def:lrs1} can be given is called the {\em{order}} of $\bu$.
For example, Fibonacci numbers are uniquely defined by the recurrence relation $f_{n+2} = f_{n+1} + f_n$ and starting elements $f_0 = 0$, $f_1 = 1$.
Note that this recurrence relation corresponds to the linear form $L(x_1,x_2)=x_1+x_2$.

We now present a second definition of linear recursive sequences which, as we will explain, is equivalent to the first definition. 
Suppose $\bu^1,\bu^2,\ldots,\bu^k$ are sequences of rationals. We say that these sequences {\em{satisfy a system of linear recurrence equations}} if there are $k$-variate linear forms $L_1,\ldots,L_k$ such that:
\begin{align}\label{def:lrs2}
\begin{cases}
u^1_{n+1} = L_1(u^1_n,\ldots,u^k_n),\\
\vdots \\
u^k_{n+1} = L_k(u^1_n,\ldots,u^k_n).
\end{cases}
\end{align}
for all $n\in \bbN$. Note that such a system can be equivalently rewritten in the matrix form
\[\vec{u}_{n+1} = M \vec{u}_n\]
where $\vec{u}_n=(u^1_n,\ldots,u^k_n)^{\transpose}$ and $M$ is the $k\times k$ matrix over $\bbQ$ such that $M\vec{x}=(L_1(\vec{x}),\ldots,L_k(\vec{x}))^{\transpose}$ for all $\vec{x}\in \bbQ^k$.
Note that then $\vec{u}_n=M^n \vec{u}_0$ for all $n\in \bbN$.

It is well known that systems of linear recurrence equations can be equivalently used to define linear recursive sequences, as explained in the following result.

\begin{proposition}[\cite{halava2005skolem}]\label{prop:lrs}
A sequence $\bu$ is a linear recursive sequence if and only if there exists $k\in \bbN$ and sequences $\bu^1,\ldots,\bu^k$ that satisfy a system of linear recurrence equations, where $\bu^1=\bu$.
Moreover, the smallest $k$ for which this holds is the order of $\bu$.
\end{proposition}

To get more accustomed with this equivalent definition, let us consider the sequence $a_n=n^2$. Since $(n+1)^2 = n^2 + 2n + 1$, we consider two auxiliary sequences $b_n = n$ and $c_n = 1$.
The initial values of these sequences are $a_0 = b_0 = 0$ and $c_0 = 1$. Thus, $a_n$ can be defined by providing these initial values together with a system of linear equations
\begin{equation}\label{eq:square}
\begin{cases}
a_{n+1} = a_n + 2b_n + c_n, \\
b_{n+1} = b_n + c_n, \\
c_{n+1} = c_n.
\end{cases}
\end{equation}
In the matrix form, we could equivalently write that
$(a_n, b_n, c_n)^\transpose = M^n \vec{e}$,
where
\begin{align*}
M = 
\begin{pmatrix}
1 & 2 & 1\\
0 & 1 & 1 \\
0 & 0 & 1
\end{pmatrix},\;
\vec{e} = 
\begin{pmatrix}
0\\
0 \\
1
\end{pmatrix}.
\end{align*}
It can be readily verified that $a_n$ is also defined by the recurrence $a_{n+3} = 3a_{n+2} - 3a_{n+1} + a_n$.

The difference between the two definitions is that in~\eqref{def:lrs1} we have only one sequence, but the depth of the recursion can be any~$k$. Conversely, in~\eqref{def:lrs2} we are allowed to have $k$ sequences, but the depth of recursion is~$1$.
The equivalence provided by Proposition~\ref{prop:lrs} is quite convenient and is often used in the literature, see \eg~\cite{OuaknineW15}.

We give a short proof of Proposition~\ref{prop:lrs}, different from the proof in~\cite{halava2005skolem}. The reason is that this proof provides us with intuition that will turn out to be useful later on.

\begin{proof}[Proof of Proposition~\ref{prop:lrs}]
 For the left-to-right implication, suppose $\bu$ is a linear recursive sequence of order $k$; say it is defined by the recursive formula $u_{n+k}=L(u_n,\ldots,u_{n+k-1})$, where $L$ is a $k$-variate linear form.
 Define the sequences $\bu^1,\ldots,\bu^k$ by setting
 \[u^i_n\coloneqq u_{n+i-1}\qquad\textrm{for all }i\in \{1,\ldots,k\}\textrm{ and }n\in \bbN.\]
 Then $\bu^1=\bu$ and the sequences $\bu^1,\ldots,\bu^k$ satisfy the system of equations as in~\eqref{def:lrs2}, where $L_k=L$ and $L_i(x_1,\ldots,x_k)=x_{i+1}$ for $i\in \{1,\ldots,k-1\}$.
 
 For the right-to-left implication, suppose that there exist $k\in \bbN$ and sequences $\bu^1,\ldots,\bu^k$ that satisfy the system of equations~\eqref{def:lrs2} for some linear forms $L_1,\ldots,L_k$,
 such that $\bu=\bu^1$. Let $M$ be a $k\times k$ matrix over $\bbQ$ that encodes the linear forms $L_1,\ldots,L_k$; that is, $\vec{u}_n=M^n \vec{u}_0$, where $\vec{u}_n=(u^1_n,\ldots,u^k_n)^\transpose\in \bbQ^k$.
 Consider the linear map $R\colon \bbQ^k\to \bbQ^{k+1}$ defined as
 \[R(\vec{x}) = (\ \vec{e} M^0 \vec{x}\, ,\, \vec{e} M^1 \vec{x}\, ,\ \ldots\ ,\,\vec{e} M^k \vec{x}\ )^\transpose,\]
 where $\vec{e}=(1,0,\ldots,0)\in \bbQ^k$.
 Note that
 \begin{align}\label{eq:map}
   R(\vec{u}_n)=(u^1_n,u^1_{n+1},\ldots,u^1_{n+k})=(u_n,u_{n+1},\ldots,u_{n+k})\qquad\textrm{for all }n\in \bbN.
 \end{align}
 Observe that $R$ is a linear map from $\bbQ^k$ to $\bbQ^{k+1}$, hence the image of $R$ is a linear subspace of $\bbQ^{k+1}$ of co-dimension at least $1$.
 Hence, there exists a nonzero linear form $K\colon \bbQ^{k+1}\to \bbQ$ such that $\mathrm{im}\, R \subseteq \ker K$, or equivalently $K(R(\vec{x}))=0$ for all $\vec{x}\in \bbQ^k$.
 By~\eqref{eq:map}, we have
 \begin{align}\label{eq:cancelling-form}
 K(u_n,u_{n+1},\ldots,u_{n+k})=0\qquad\textrm{for all }n\in \bbN.
 \end{align}
 Let $a_0,a_1,\ldots,a_k\in \bbQ$ be such that
 \[K(x_0,\ldots,x_k)=a_0x_0+\ldots+a_kx_k.\]
Since $K$ is nonzero there exists the largest index $t$ such that $a_t\neq 0$.
From~\eqref{eq:cancelling-form} we infer that
 \[u_{n+t}= - \frac{a_{t-1}}{a_t}\cdot u_{n+t-1} - \frac{a_{t-2}}{a_t}\cdot u_{n+t-2} - \ldots -\frac{a_{0}}{a_t}\cdot u_{n} \qquad\textrm{for all }n\in \bbN,\]
 so $\bu$ is a linear recursive sequence of order at most $t$.
\end{proof}

\begin{remark}\label{rem:homog}
One could imagine setting up all the definitions presented above using {\em{affine forms}} instead of linear forms, that is, functions $A\colon \bbQ^k\to \bbQ$ of the form
\[A(x_1,\ldots,x_k)=a_1x_1+\ldots+a_kx_2+c,\]
where $a_1,\ldots,a_k,c\in \bbQ$. However, as we may always add constant sequences to the system of recurrence equations defining a sequence,
considering affine forms does not increase the expressive power. In fact, 
from Proposition~\ref{prop:lrs} it can be easily derived that we obtain exactly the same class of linear recursive sequences, regardless of whether we use linear or affine forms in both definitions.
\end{remark}

\subparagraph*{Poly-recursive sequences.} 
We now generalise the concept of linear recursive sequences by allowing polynomial functions instead of only linear forms.
The starting point of the generalisation is the definition via a system of recurrence equations, as in~\eqref{def:lrs2}.

\begin{definition}\label{deff:prs}
 A sequence of rationals $\bu$ is {\em{polynomial recursive}} (or {\em{poly-recursive}} for short) if there exist $k\in \bbN$, sequences of rationals $\bu^1,\ldots,\bu^k$ satisfying $\bu=\bu^1$, 
 and polynomials $P_1,\ldots,P_k\in \bbQ[x_1,\ldots,x_k]$ such that for all $n\in \bbN$, we have
\begin{align}\label{def:prs}
\begin{cases}
u^1_{n+1} = P_1(u^1_n,\ldots,u^k_n),\\
\vdots \\
u^k_{n+1} = P_k(u^1_n,\ldots,u^k_n).
\end{cases}
\end{align}
\end{definition}
Again, notice that polynomials $P_1, \ldots, P_k$ and the initial values $u^1_0, \ldots u^k_0$ uniquely determine the sequences $\bu^1,\ldots,\bu^k$, hence in particular the sequence $\bu=\bu^1$.


Let us examine a few examples. First, recall the sequences $a_n = 2^{2^n}$ and $b_n = n!$ defined in Section~\ref{sec:introduction}. 
Another example is the sequence $d_n = 2^{n^2}$.  Since $2^{(n+1)^2} = 2^{n^2 +
  2n + 1}$, we define $d_0 = e_0 = 1$ and let
\begin{equation*}
\begin{cases}
d_{n+1} = d_n \cdot (e_n)^2 \cdot 2, \\
e_{n+1} = e_n \cdot 2.
\end{cases}
\end{equation*}
The polynomials used in the last definition are $P_1(x_1,x_2) = 2x_1(x_2)^2$ and $P_2(x_1,x_2) = 2x_2$.
Notice that this idea can be easily generalised to define any sequence of the form $r^{Q(n)}$, where $r$ is a rational number and $Q$ is a polynomial with rational coefficients.
We remark that all three sequences $a_n=2^{2^n}$, $b_n=n!$, $d_n = 2^{n^2}$ are not linear recursive for simple asymptotic reasons (from the discussion in Section~\ref{sec:introduction}). 

\section{Simple poly-recursive sequences}
\label{sec:simple}
The following notion is a natural generalisation of the definition~\eqref{def:lrs1} of linear recursive sequences to the setting of recurrences defined using polynomials.

\begin{definition}\label{def:simple-prs}
 A sequence of rationals $\bu$ is {\em{simple poly-recursive}} if there exists $k\in \N$ and a polynomial $P\in \bbQ[x_1,x_2,\ldots,x_k]$ such that
 \begin{align}\label{def:prs2}
u_{n+k} = P(u_n,u_{n+1},\ldots,u_{n+k-1})\qquad\textrm{for all }n\in \N.  
 \end{align}
\end{definition}

Again, note that if $\bu$ is simple poly-recursive as above, then the polynomial $P$ and the first $k$ values $u_0,\ldots,u_{k-1}$ uniquely determine the sequence $\bu$.

Clearly, every linear recursive sequence is a simple poly-recursive sequence.
In fact, by Proposition~\ref{prop:lrs} and Remark~\ref{rem:homog}, the two notions would coincide if we required that the polynomial $P$ in the definition above has degree at most $1$.
On the other hand, observe that the same construction as in the first paragraph of the proof of Proposition~\ref{prop:lrs} shows that every simple poly-recursive sequence is poly-recursive.
We now prove that this inclusion is strict.




\begin{theorem}\label{prop:factorial}
The sequence $b_n = n!$ is not simple poly-recursive.
\end{theorem}
\begin{proof}
Towards a contradiction, suppose there is $k\in \N$ and a polynomial $P\in \bbQ[x_1,\ldots,x_k]$ such that
\begin{align}\label{eq:bpoly}
b_{n+k}=P(b_n,b_{n+1},\ldots,b_{n+k-1})\qquad\textrm{for all }n\in \N. 
\end{align}
Let us write
\[P=Q+A,\]
where $Q,A\in \bbQ[x_1,\ldots,x_k]$ are such that $A$ is the sum of all the monomials in the expansion of $P$ that have degree at most $1$, while $Q$ is the sum of all the remaining monomials in the expansion of $P$.
Thus, $A$ is an affine form, while every monomial in the expansion of $Q$ has total degree at least $2$.

Since $A$ is an affine form, there exists a number $c\in \N$ such that
\[|A(q_1,\ldots,q_k)|<c+c\cdot \max_{1\leq i\leq k} |q_i|\qquad \textrm{for all }q_1,\ldots,q_k\in \bbQ.\]
Thus, for all $n>2c$ we have
\begin{align}\label{eq:Asmall}
|A(b_n,b_{n+1},\ldots,b_{n+k-1})| \leq c+c\cdot (n+k-1)!<(n+k)!=b_{n+k}.
\end{align}
Since by~\eqref{eq:bpoly} it follows that
\[Q(b_n,b_{n+1},\ldots,b_{n+k-1}) = b_{n+k}-A(b_n,b_{n+1},\ldots,b_{n+k-1}),\]
using~\eqref{eq:Asmall} we may conclude that for all $n>2c$ the following inequality holds:
\begin{align}\label{eq:small} 
0 < Q(b_n,b_{n+1},\ldots,b_{n+k-1}) < 2b_{n+k}.
\end{align}

Let $m$ be the product of all denominators of all the coefficients appearing in the expansion of $P$ into a sum of monomials.
Note that for all $n>m$, the number $\td{b}_n\coloneqq \frac{b_n}{m}=\frac{n!}{m}$ is an integer. Furthermore, we have that $\td{b}_n$ divides $\td{b}_{n'}$ for all $n'\geq n$.
Since every monomial in the expansion of $Q$ has total degree at least $2$, we infer that for all $n>m$, we have
\begin{align}\label{eq:divisible}
 \left(\td{b}_n\right)^2\ |\ Q(b_n,b_{n+1},\ldots,b_{n+k-1}). 
\end{align}
By combining~\eqref{eq:divisible} with the left inequality of~\eqref{eq:small}, we conclude that for all $n>\max(2c,m)$,
\[ Q(b_n,b_{n+1},\ldots,b_{n+k-1}) \geq \left(\td{b}_n\right)^2.\]
This bound together with the right inequality of~\eqref{eq:small} implies that
\[\left(\frac{n!}{m}\right)^2 = \left(\td{b}_n\right)^2 < 2b_{n+k} = 2\cdot (n+k)!.\]
This inequality, however, is not true for every sufficiently large $n$, a contradiction.
\end{proof}

\section{Modular periodicity}
\label{sec:periodicity}
Recall that a sequence of numbers $\br$ is {\em{ultimately periodic}} if there exist $N,k\in \N$ such that for all $n\geq N$, we have $r_n=r_{n+k}$.
In this section we prove the following periodicity property of poly-recursive sequences, which, by means of contradiction, provides a basic technique for proving that a given sequence is not poly-recursive.

\begin{theorem}\label{thm:periodic}
 Suppose $\bu$ is a poly-recursive sequence of integers. Then there exists $a\in \N$ such that for every prime $p>a$, the sequence $r_n\coloneqq u_n\bmod p$ is ultimately periodic.
\end{theorem}
\begin{proof}
 Let $\bu$ be defined by the system of recursive equations
 \begin{align}\label{def:prs-periodic}
\begin{cases}
u^1_{n+1} = P_1(u^1_n,\ldots,u^k_n),\\
\vdots \\
u^k_{n+1} = P_k(u^1_n,\ldots,u^k_n),
\end{cases}
\end{align}
where $\bu^1,\ldots,\bu^k$ are sequences such that $\bu^1=\bu$ and $P_1,\ldots,P_k\in \bbQ[x_1,\ldots,x_k]$.

\newcommand{\wt}[1]{\widetilde{#1}}

Without loss of generality we may assume that the initial values $u_0^1,\ldots,u_0^k$ are integers.
Indeed, this is certainly the case for $u_0^1=u_0$, while for every $i>1$, 
we may rewrite the system so that it uses the sequence $\wt{\bu}^i=q_i\cdot \bu^i$ instead of
$\bu^i$, where $q_i$ is the denominator of $u^i_0$. 
For this, the starting condition for $\wt{\bu}^i$ can be set as $\wt{u}^i_0=q_i\cdot u^i_0$, which is an integer, 
in all polynomials $P_1,\ldots,P_k$ we may substitute $x_i$ with $x_i/q_i$, and the polynomial $P_i$ can be replaced with $q_i\cdot P_i$.

Further, without loss of generality we may assume that all the monomials present in the expansions of all the polynomials $P_1,\ldots,P_k$ have the same total degree $d>1$.
Indeed, let $d>1$ be any integer that is not smaller than the degrees of all the polynomials $P_1,\ldots,P_k$. To the system~\eqref{def:prs-periodic} we add a new sequence $\bu^{k+1}$, defined by
setting 
\[u^{k+1}_0=1\qquad\textrm{and}\qquad u^{k+1}_{n+1}=\left(u^{k+1}_n\right)^d\ \textrm{for }n\in \bbN.\]
Thus $\bu^{k+1}$ is constantly equal to $1$. 
Then each monomial $M(x_1,\ldots,x_k)$ appearing in the expansion of any of the polynomials $P_i(x_1,\ldots,x_k)$ can be replaced by 
the monomial $M(x_1,\ldots,x_k)\cdot x_{k+1}^{d-t}\in \bbQ[x_1,\ldots,x_k,x_{k+1}]$, where $t$ is the total degree of $M$. 
It is straightforward to see that the modified system of recursive equations still defines $\bu=\bu^1$, while all monomials appearing in all the polynomials used in it have the same degree $d$.

After establishing these two assumptions, we proceed to the main proof. Let $a\in \N$ be a positive integer such that the polynomials 
\[\td{P}_i\coloneqq a\cdot P_i\]
all belong to $\bbZ[x_1,\ldots,x_k]$, that is, have integer coefficients.
For instance, one can take $a$ to be product of all the denominators of all the rational coefficients appearing in the 
polynomials $P_1,\ldots,P_k$. For all $i\in \{1,\ldots,k\}$ and $n\in \N$, let us define 
\[\td{u}^i_n\coloneqq a^{\frac{d^n-1}{d-1}}\cdot u^i_n.\]
By a straightforward induction we show that the sequences $\td{\bu}^1,\ldots,\td{\bu}^k$ satisfy the system of recursive equations
\begin{align}\label{def:prs-periodic2}
\begin{cases}
\td{u}^1_{n+1} = \td{P}_1(\td{u}^1_n,\ldots,\td{u}^k_n),\\
\vdots \\
\td{u}^k_{n+1} = \td{P}_k(\td{u}^1_n,\ldots,\td{u}^k_n).
\end{cases}
\end{align}
Indeed, the induction base is trivial and for the induction step recall that all monomials have the same degree $d$, hence
$$
\td{P}_i(\td{u}^1_n,\ldots,\td{u}^k_n) = a\cdot P_i(a^{\frac{d^n-1}{d-1}}\cdot u^1_n,\ldots,a^{\frac{d^n-1}{d-1}}\cdot u^k_n) = a\cdot a^{\frac{d^{n+1} - d}{d-1}}\cdot u^i_{n+1} = a^{\frac{d^{n+1} - 1}{d-1}}\cdot u^i_{n+1} =\td{u}^i_{n+1}.
$$
Observe that since the initial values $\td{u}^i_0=u^i_0$ are integers, and the polynomials $\td{P}_i$ have integer coefficients,
we can infer that all the entries of sequences $\td{\bu}^1,\ldots,\td{\bu}^k$ are integers.

We now show that for every prime $p>a$, the sequence $\br$ defined as
$r_n=u_n\bmod p$ is ultimately periodic; this will conclude the proof.  
For every $i\in \{1,\ldots,k\}$ and $n\in \N$, let 
\[\td{r}^i_n\coloneqq \td{u}^i_n\bmod p.\]
By~\eqref{def:prs-periodic2} and the fact that polynomials $\td{P}_i$ have integer coefficients, 
for every $n\in \N$ the vector of entries $(\td{r}^1_{n+1},\ldots,\td{r}^k_{n+1})$ is uniquely determined by the vector $(\td{r}^1_{n},\ldots,\td{r}^k_{n})$.
Since this vector may take only at most $p^k$ different values, it follows that the sequences $\td{\br}^1,\ldots,\td{\br}^k$ are ultimately periodic.

Now note that for every $n\in \N$, we have
\[a^{\frac{d^n-1}{d-1}}\cdot r_n \equiv a^{\frac{d^n-1}{d-1}} \cdot u_n = \td{u}^1_n \equiv \td{r}^1_n \mod p.\]
Since $p>a$ and $p$ is a prime, we have that $a$ and $p$ are coprime. 
Therefore, there exists an integer $b$ such that $ab\equiv 1\bmod p$. By multiplying the above congruence by $b^{\frac{d^n-1}{d-1}}$, we have
\begin{align}\label{eq:rn-final}
r_n\equiv b^{\frac{d^n-1}{d-1}}\cdot \td{r}^1_n \mod p.
 \end{align}
Observe that the sequence $b_n=b^{\frac{d^n-1}{d-1}}$ satisfies the recursive equation $b_{n+1}=b\cdot \left(b_n\right)^d$, hence the sequence $(b_n \bmod p)$ is ultimately periodic.
Since $\td{\br}^1$ is ultimately periodic as well, from~\eqref{eq:rn-final} we conclude that the sequence $\br$ is ultimately periodic.
\end{proof}

We use Theorem~\ref{thm:periodic} to prove that
the Catalan numbers are not poly-recursive.
Recall that the $n$th Catalan number $C_n$ is given by the formula
$C_n=\frac{1}{n+1} \binom{2n}{n}$.

Alter and Kubota~\cite{klapki} studied the behaviour of the Catalan numbers modulo primes.
It is easy to see (and proved in~\cite{klapki}) that for every prime $p$, the sequence $C_n$ contains infinitely many numbers divisible by $p$, and infinitely many numbers not divisible by $p$.
Let a {\em{$p$-block}} be a maximal contiguous subsequence of the sequence $C_n$ consisting of entries divisible by $p$.
The $p$-blocks can be naturally ordered along the sequence $C_n$, so let $L_k^p$ be the length of the $k$th $p$-block.
Then Alter and Kubota proved the following.

\begin{theorem}[\cite{klapki}]\label{thm:klapki}
 For every prime $p>3$ and $k\geq 1$, we have
 $$L_k^p=\frac{p^{m+1}-3}{2},$$
 where $m$ is the largest integer such that $\left(\frac{p+1}{2}\right)^m$ divides $k$.
\end{theorem}

Note that Theorem~\ref{thm:klapki} in particular implies that for every prime $p>3$, the sequence $C_n$ contains arbitrary long $p$-blocks.
This means that $C_n$ taken modulo $p$ cannot be ultimately periodic. By combining this with Theorem~\ref{thm:periodic}, we conclude the following.

\begin{corollary}\label{cor:catalan}
 Catalan numbers are not poly-recursive.
\end{corollary}

\section{Cancelling polynomials}
\label{sec:cancelling}
Consider the following definition, which can be seen as a variation of the definition of simple poly-recursive sequences, which we discussed in Section~\ref{sec:simple}.

\begin{definition}
  A sequence of rationals $\bu$ admits a \emph{cancelling polynomial} if
  there exist $k\in \N$ and a nonzero polynomial \(P \in \bbQ[x_0, \ldots, x_k]\) such that
\[P\left(u_{n}, u_{n+1}, \ldots, u_{n+k} \right) = 0\qquad \textrm{for all }n \in \bbN.\]
\end{definition}

\begin{remark}\label{rem:cancelling}
A cancelling polynomial $P$ can be always assumed to have integer coefficients, i.e. to belong to $\bbZ[x_0,\ldots,x_k]$, because
one may multiply $P$ by the product of all denominators that occur in its coefficients.
\end{remark}

Observe that the notion of a cancelling polynomial extends the definition of simple poly-recursive sequences (Definition~\ref{def:simple-prs}) in the following sense:
a sequence is simple poly-recursive if and only if it admits a cancelling polynomial $P(x_0,\ldots,x_k)$ whose expansion into a sum of monomials involves only one term containing $x_k$, namely the monomial $x_k$ itself.
This particular form of the considered algebraic constraint was vitally used in the proof of Proposition~\ref{prop:factorial}, where we showed that the sequence $b_n=n!$ is not simple poly-recursive.
In fact, if one drops this restriction, then it is easy to see that the sequence $b_n=n!$ actually admits a cancelling polynomial: for instance $P(x_0,x_1,x_2) = x_0x_2 - (x_1)^2 - x_0x_1$.




We now prove that the above example is not a coincidence.

\begin{theorem}\label{thm:cancelling}
Every poly-recursive sequence admits a cancelling polynomial.
\end{theorem}
\begin{proof}
The proof follows the same basic idea as the proof of Proposition~\ref{prop:lrs} that we gave in Section~\ref{sec:preliminaries}.
The difference is that instead of linear maps we work with maps defined by polynomial functions, hence instead of linear independence we shall work with the notion of algebraic independence.

Recall that if $\bbK\subseteq \bbL$ is a field extension, then elements $a_1,\ldots,a_k\in \bbL$ are {\em{algebraically dependent}} over $\bbK$ if there is a nonzero polynomial $P\in \bbK[x_1,\ldots,x_k]$ such that
$P(a_1,\ldots,a_k)=0$ in $\bbL$. We will use the following well-known fact; see e.g.~\cite[Chapter~VIII, Theorem~1.1]{Lang}.

\begin{claim}\label{cl:independent}
 If $\bbK$ is a field and $k\in \mathbb{N}$, then in the field of rational expressions $\bbK(x_1,\ldots,x_k)$ every $k+1$ elements are algebraically dependent over $\bbK$.
\end{claim}

We proceed to the proof of the theorem. Let $\bu$ be the poly-recursive sequence in question. By definition, for some $k\in \mathbb{N}$ there are sequences $\bu^1,\ldots,\bu^k$ and polynomials $P_1,\ldots,P_k\in \bbQ[x_1,\ldots,x_k]$
such that for all $n\in \mathbb{N}$, 
\begin{align*}
\begin{cases}
u^1_{n+1} = P_1(u^1_n,\ldots,u^k_n),\\
\vdots \\
u^k_{n+1} = P_k(u^1_n,\ldots,u^k_n).
\end{cases}
\end{align*}
We inductively define polynomials $P^{(t)}_1,\ldots,P^{(t)}_k\in \bbQ[x_1,\ldots,x_k]$ as follows.
For $t=0$, set
\begin{align*}
 P^{(0)}_i(x_1,\ldots,x_k) = x_i\qquad \textrm{for all }i\in \{1,\ldots,k\},
\end{align*}
and for $t\geq 1$, set
\begin{align*}
 P^{(t)}_i(x_1,\ldots,x_k) = P_i(P^{(t-1)}_1(x_1,\ldots,x_k),\ldots,P^{(t-1)}_k(x_1,\ldots,x_k))\quad \textrm{for all }i\in \{1,\ldots,k\}.
\end{align*}
The following claim follows from the construction by a straightforward induction.

\begin{claim}\label{cl:iteration}
 For all $n,t\in \mathbb{N}$ and $i\in \{1,\ldots,k\}$, we have $P^{(t)}_i(u^1_n,\ldots,u^k_n)=u^i_{n+t}$.
\end{claim}

Consider the polynomials 
\begin{align*}
P^{(0)}_1,P^{(1)}_1,\ldots,P^{(k)}_1\in \bbQ[x_1,\ldots,x_k].
\end{align*}
By Claim~\ref{cl:independent}, these polynomials (treated as elements of $\bbQ(x_1,\ldots,x_k)$) are algebraically dependent over $\bbQ$, so there exists a nonzero polynomial $Q\in \bbQ[y_0,y_1,\ldots,y_k]$ such that
the polynomial
$$R(x_1,\ldots,x_k)=Q(P^{(0)}_1(x_1,\ldots,x_k),P^{(1)}_1(x_1,\ldots,x_k),\ldots,P^{(k)}_1(x_1,\ldots,x_k))$$
is identically zero. It now remains to observe that by Claim~\ref{cl:iteration}, for every $n\in \mathbb{N}$ we have
$$0=R(u^1_n,\ldots,u^k_n)=Q(u^1_n,u^1_{n+1},\ldots,u^1_{n+k})=Q(u_n,u_{n+1},\ldots,u_{n+k}),$$
hence $Q$ is a cancelling polynomial for $\bu$.
\end{proof}

\begin{remark}
Notice that a given polynomial can be the cancelling polynomial of many different sequences.
For example, the polynomial $(x_0)^2 - 1$ is a cancelling polynomial of any sequence over $\set{-1,1}$. 
In particular, some of those sequences are not ultimately periodic modulo $p$, for any prime numbers $p$, and thus are not poly-recursive by Theorem~\ref{thm:periodic}. 
Hence, the converse direction of Theorem~\ref{thm:cancelling} does not hold.
\end{remark}

We now present an application of Theorem~\ref{thm:cancelling} by showing that
the sequence $u_n=n^n$ is not poly-recursive.
By Theorem~\ref{thm:cancelling}, it suffices to show that there is no cancelling polynomial for this sequence.
Contrary to the reasoning presented in Section~\ref{sec:periodicity}, where we used off-the-shelf results about modular (non)periodicity of Catalan numbers,
proving the nonexistence of a cancelling polynomial for the $n^n$ sequence turns out to be a somewhat challenging task.

We first observe that when we apply a multivariate polynomial to consecutive entries of $u_n$, we can rewrite the result in another form:

\begin{lemma}\label{lem:zero-sum}
 Let $Z\in \bbZ[x_0,x_1,\ldots,x_k]$ be a nonzero polynomial.
 Then there exist nonzero polynomials $P_1,\ldots,P_m,Q_1,\ldots,Q_m\in \bbZ[x]$ such that
 the polynomials $P_1,\ldots,P_m$ are pairwise different and for every $n\in \bbN$, 
 \[Z\left(n^n,(n+1)^{n+1},\ldots,(n+k)^{n+k}\right) = \sum_{i=1}^m P_i(n)^n\cdot Q_i(n).\]
\end{lemma}
\begin{proof}
  By expanding $Z$ as a sum of monomials, we may write
  \begin{align}\label{eq:raccoon}
    Z(x_0,\ldots,x_k) = \sum_{i = 1}^{m} c_i\cdot M_i(x_0,\ldots,x_j),   
  \end{align}
  where for all $i\in \{1,\ldots,m\}$, $c_i\neq 0$ and
  \[
    M_i(x_0,\ldots,x_k) = \prod_{j=0}^{k} x_j^{d_{i,j}}   
  \]
  are pairwise different monomials.
  Now observe that for every $n\in \bbN$, we have
  \begin{align}
    M_i\left(n^n,(n+1)^{n+1},\ldots,(n+k)^{n+k}\right) & = \prod_{j=0}^k (n+j)^{d_{i,j}\cdot (n+j)}\nonumber\\
    & = \left(\prod_{j=0}^k (n+j)^{d_{i,j}}\right)^n\cdot \prod_{j=0}^k (n+j)^{d_{i,j}\cdot j}.\label{eq:beaver}
  \end{align}
  Hence, if we define
  \[P_i(x)=\prod_{j=0}^k (x+j)^{d_{i,j}}\qquad\textrm{and}\qquad Q_i(x)=c_i\cdot \prod_{j=0}^k (x+j)^{d_{i,j}\cdot j},\]
  then, by~\eqref{eq:raccoon} and~\eqref{eq:beaver}, we conclude that 
  \[Z\left(n^n,(n+1)^{n+1},\ldots,(n+k)^{n+k}\right) = \sum_{i=1}^m P_i(n)^n\cdot
  Q_i(n)\qquad\textrm{for all }n\in \bbN,\] as required. It now suffices to observe that
  (1) all polynomials $P_i$ and $Q_i$ are nonzero, because $c_i\neq 0$ and the
  monomial $M_i$ is nonzero, and (2) the polynomials $P_i$ are pairwise
  different, because they have pairwise different multisets of roots.
\end{proof}

With \Cref{lem:zero-sum} established, we move to the main result of this section.

\begin{theorem}\label{thm:nton}
  The sequence $u_n=n^n$ is not poly-recursive.  
\end{theorem}
\begin{proof}
    Suppose, for the sake of contradiction, that the sequence $u_n=n^n$ is poly-recursive.
    By \Cref{thm:cancelling} and Remark~\ref{rem:cancelling}, there exists a nonzero polynomial $Z\in \bbZ[x_0,x_1,\ldots,x_k]$ that is cancelling for $u_n$.
    By \Cref{lem:zero-sum}, we can find nonzero polynomials $P_1,\ldots,P_m,Q_1,\ldots,Q_m\in \bbZ[x]$, where $P_1,\ldots,P_m$ are pairwise different, such that
    \begin{align}\label{eq:badger}\sum_{i=1}^m P_i(n)^n\cdot Q_i(n) = 0 \qquad\textrm{for all }n\in \bbN.\end{align}
    This system of equations seems somewhat unwieldy due to the presence of the term $P_i(n)^n$, where $n$ is involved both in the base and in the exponent.
    The following claim formulates the key idea of the proof: if we consider the equations~\eqref{eq:badger} modulo any prime, then the bases and the exponents of these terms can be made independent.

    \begin{claim}\label{cl:system-modp}
      For every prime $p$ and all \(a,b\in \bbZ\) where $b>0$, it holds that
      \[\sum_{i=1}^m P_i(a)^b\cdot Q_i(a) \equiv 0 \mod p\enspace.\]
    \end{claim}
    \begin{claimproof}
      Since \(p\) and \(p-1\) are coprime, there is an \(n \in \bbZ\) such that
      \(n\equiv a\bmod p\) and \(n\equiv b\bmod p-1\).  Thus for any \(1 \le i \le m\):
      \[Q_i(n)\equiv Q_i(a)\mod p \qquad\text{and}\qquad P_i(n)^n\equiv P_i(a)^n\equiv P_i(a)^b\mod
      p\enspace,\]%
      the second part holding by Fermat's Little Theorem. The claim now follows by considering equality~\eqref{eq:badger} modulo $p$.
    \end{claimproof}
    
  Let \(a \in \bbN\) and let \(D_a=[d_{ij}]_{1\leq i,j\leq m}\) be the
  \(m\times m\) matrix defined by \(d_{ij} = P_j(a)^i\). Since this is essentially a Vandermonde matrix, its determinant has a simple expression, 
  as expressed in the following claim.
  
    \begin{claim}\label{cl:det-Da}
      Let \(S \in \bbZ[x]\) be defined as
      \[S(x) = \prod_{i=1}^m P_i(x)\cdot \prod_{1\leq i<j\leq m} (P_i(x)-P_j(x))\enspace.\] Then \(S\) is nonzero and \(\det (D_a) = S(a)\).
    \end{claim}
    \begin{claimproof}
      That \(S\) is nonzero follows from the fact that the polynomials \(P_i\) are
      all nonzero and pairwise different.

      Now observe that \(D_a\) is a Vandermonde matrix with columns consisting
      of consecutive powers of \(P_j(a)\), for \(1 \le j \le m\), where additionally every
      \(j\)th column is multiplied by \(P_j(a)\).  It is well known that the
      determinant of the Vandermonde matrix \([P_j(a)^{i-1}]_{1\leq i,j\leq m}\) is
      \[
      \prod_{1\leq i<j\leq m} (P_i(a)-P_j(a))\enspace.
      \]
      Further, multiplying the $j$th column by \(P_j(a)\), for all $j$, results in the determinant
      being multiplied by \(\prod_{i=1}^m P_i(a)\).  This proves the claim.
    \end{claimproof}

  We will need the following classical definition.
  \begin{definition}\label{cl:adjugate}
    Let \(R\) be a ring and \(M\) be a \(m\times m\) matrix over \(R\).  The \emph{adjugate
      matrix} \(\wh{M}\) of~\(M\) is the \(m\times m\) matrix over \(R\) that satisfies
    \(\wh{M} M = \det(M)\cdot I\), where \(I\) is the \(m\times m\) identity matrix.
  \end{definition}
  It is well known that the adjugate matrix always exists.
  Now let \(u_a=(Q_1(a),\ldots,Q_m(a))^{\transpose}\).  \Cref{cl:system-modp} implies that for every
  prime \(p\), \[D_a u_a \equiv \vec{0} \mod p,\] where $\vec{0}$ is the $m$-dimensional zero vector.  By multiplying both sides of this
  equation by the adjugate matrix of \(D_a\) taken over \(\bbZ_p\), we conclude that
  for every prime \(p\), we have
  \[\det(D_a)\cdot u_a \equiv \vec{0} \mod p\qquad\textrm{for all }a\in \bbN.\]
  This is equivalent to
  \begin{align}\label{eq:final}
  S(a)\cdot Q_i(a)\equiv 0 \mod p\qquad \text{for all }a\in \bbN\text{ and
  }1\le i \le m.
  \end{align}
  This means that for every prime \(p\) and every \(1 \le i \le m\), the following assertion holds: every
  \(a\in \bbF_p\) is a zero of the polynomial \(S\cdot Q_i\) considered as a polynomial over \(\bbF_p\).
  
  Recall that the polynomials $S,Q_1,\ldots,Q_m\in \bbZ[x]$ are nonzero.
  Consider a prime \(p\) that is larger than every coefficient occurring in the
  expansion of the polynomials \(S\),
  \(Q_1, \ldots, Q_m\) into sums of monomials, and that is further larger than
  \(\deg(S) + \max_{j \in \set{1,\ldots, m}} \deg(Q_j)\). Then the polynomials \(S,Q_1,\ldots,Q_m\) are nonzero even when regarded as polynomials over \(\bbF_p\), hence
  the same can be said also about the polynomials $S\cdot Q_i$, for all $1\leq i\leq m$.
  However, by~\eqref{eq:final}, for every \(1 \le i \le m\) the polynomial \(S\cdot Q_i\) has at least
  \(p > \deg(S)+\deg(Q_i)\) roots over \(\bbF_p\). This is a contradiction. 
\end{proof}

\section{Applications in weighted automata}
\label{sec:weighted}
In this section we discuss the implications of the results we presented in the previous sections for various questions regarding the expressive power of extensions of weighted automata. We will briefly describe the model of weighted automata and focus only on its expressive power. We refer an interested reader to \eg~\cite{AlmagorBK11,DrosteHWA09} for an introduction to the area.

Given a semiring $\bbS$, a \emph{weighted automaton} $\cA$ is a tuple $(d,\Sigma,\set{M_a}_{a \in \Sigma},I,F)$, where: 
\begin{itemize}
\item $d\in \N$ is the dimension;
\item $\Sigma$ is a finite alphabet;
\item every $M_a$ is a $d\times d$ matrix over $\bbS$; and
\item $I$ and $F$ are the initial and the final vector in $\bbS^d$, respectively.
\end{itemize}
In this paper we only consider the semiring $\bbS=\bbQ$.
A weighted automaton defines a function $\sem{\cA} \colon \Sigma^* \to \bbS$ as follows: if $w = a_1 \ldots a_n \in \Sigma^*$, then
\begin{align}\label{def:wa}
\sem{\cA}(w) \; = \; I^{\transpose} \cdot M_{a_1}M_{a_2}\ldots M_{a_n} \cdot F.
\end{align}
Note that when $|\Sigma| = 1$, this definition coincides with (the matrix form of) the definition~\eqref{def:lrs2} of linear recursive sequences. 
Assuming $|\Sigma| = 1$, one can identify each word with its length, which means that a weighted automaton defines a sequence $\sem{\cA} \colon \bbN \to \bbS$. 
Therefore, weighted automata recognise exactly linear recursive sequences. See~\cite{BarloyFLM20} for a broader discussion of the connection between linear recursive sequences and weighted automata.

We now discuss three nonlinear extensions of weighted automata that can be found in the literature. These extensions are studied in different areas and, as far as we are aware, 
they have never been compared in terms of expressive power before. 
We show that the results we presented in Sections~\ref{sec:periodicity} and~\ref{sec:cancelling} can be used to prove separation results, in terms of the expressive power, for some of these classes. 

Like in the case of weighted automata, any automaton within the considered classes defines a function $f \colon \Sigma^* \to \bbQ$, where $\Sigma$ is the working alphabet. 
For our purposes, we restrict attention to the case of unary alphabets, that is, $|\Sigma| = 1$. Thus, the three considered classes of extended weighted automata correspond to three separate classes of sequences $f \colon \bbN \to \bbQ$, similarly as standard weighted automata correspond to the class of linear recursive sequences.

\subparagraph*{Cost-register automata (CRA).}
Cost-register automata (CRA) were introduced in at least three contexts~\cite{Senizergues07,AlurDDRY13,BenediktDSW17}. To avoid technical details, we simply observe that CRAs over unary alphabets recognize exactly poly-recursive sequences, as defined in Definition~\ref{deff:prs}. Since~\cite{Senizergues07,AlurDDRY13,BenediktDSW17} discuss several variants of CRAs, to avoid ambiguity we refer to the definition of a CRA that can be found in~\cite{MazowieckiR19}\footnote{The equivalence of CRAs and poly-recursive sequences over a unary alphabet is basically a syntactic translation, if one assumes that CRAs have only one state. Proving that every CRA can be defined by a one state CRA is a simple encoding of states into the registers.}.

\subparagraph*{Weighted context-free grammars (WCFG).}
Weighted automata can be equivalently defined as an extension of finite automata, where each translation is labelled by an element of the semiring $\bbS$ (see \eg~\cite{AlmagorBK11}). In short, each run is assigned a value: the semiring product of the labels of all the transitions used in the run. Given a word $w$, the automaton outputs the semiring sum of the values assigned to all runs accepting $w$.

Weighted context-free grammars are an extension of context-free grammars in the same way weighted automata are an extension of finite automata. Every grammar rule is assigned a label from $\bbS$. Then every derivation tree is assigned the semiring product of the labels of all the rules used in the tree. The output for a word $w$ is defined as the semiring sum of all values assigned to derivation trees of $w$. See \eg~\cite{GantyG18} for more details. Here we present only one example from~\cite{GantyG18} over the semiring $\bbQ$.

Consider the grammar with one nonterminal $X$ (which is also the starting nonterminal) and one terminal $a$ with the following rules: $X \to a$, $X \to aXX$. Both rules are assigned weight~$1$. Therefore, for every word $a^n$ the output is the number of derivation trees. It is easy to see that if we denote the output on the word $a^n$ by $C_n$, then $C_0 = 1$ and $C_{n+1} = \sum_{i = 0}^{n} C_iC_{n-i}$ for all $n\in \N$, hence
$C_n$ is just the $n$th Catalan number. This proves that Catalan numbers can be defined by a unary-alphabet WCFG over $\bbQ$.
By Corollary~\ref{cor:catalan}, we can now conclude the following.

\begin{corollary}
The class of sequences definable by unary-alphabet WCFGs over $\bbQ$ is not contained in the class of sequences recognizable by unary-alphabet CRAs over $\bbQ$.
\end{corollary}

\subparagraph*{Weighted MSO (WMSO).}
Weighted MSO logic~\cite{DrosteG07,KreutzerR13} was introduced as a logic involving weights that intended to capture the expressive power of weighted automata, similarly as finite automata are characterized by MSO. In general, WMSO turns out to be strictly more expressive than weighted automata. We will not define the whole syntax of WMSO, only a simple fragment that does not even use variables. See~\cite{DrosteG07,KreutzerR13} for the full definition.

Fix the semiring $\bbS = \bbQ$.
Similarly as for weighted automata, every WMSO formula $\varphi$ over $\bbQ$ defines a function $\sem{\varphi}\colon \Sigma^*\to \bbQ$.
As for atomic formulas, every $c \in \bbQ$ is an atomic formula that defines the constant function $\sem{c}(w) = c$. Instead of the boolean connectives $\vee$ and $\wedge$, 
WMSO formulas can be added using $+$ and multiplied using $\cdot$, with the obvious semantics. 
Instead of having the existential quantifier $\exists_x$ and the universal quantifier $\forall_x$,
we have the sum quantifier $\sum_x$ and the product quantifier $\prod_x$. 
Then \[\sem{\sum_x \varphi} (w) = \sum_{i=1}^n \sem{\varphi[x \to a_i]}(w)\qquad\textrm{for all }w = a_1\ldots a_n\in \Sigma^*,\] 
and similarly for $\sem{\prod_x \varphi} (w)$. For example, $\sem{\sum_x 1} (a^n) = n$.
It follows that \[\sem{\prod_x \sum_y 1} (a^n) = n^n.\]
This proves that the sequence $n^n$ can be defined in unary-alphabet WMSO over $\bbQ$,
so by Theorem~\ref{thm:nton} we may conclude the following.

\begin{corollary}
The class of sequences definable in unary-alphabet WMSO over $\bbQ$ is not contained in the class of sequences recognizable by unary-alphabet CRAs over $\bbQ$.
\end{corollary}

\section{Conclusion}
\label{sec:conclusion}
We proved that two sequences, the Catalan numbers $C_n$ and $u_n = n^n$, are not
polynomial recursive.  For this, we exhibited two properties that poly-recursive
sequences always satisfy: ultimate periodicity modulo large prime numbers and
admitting a cancelling polynomial.

Going further than poly-recursive sequences, one can consider the class of
{\em{rational recursive sequences}}.  These are specified like polynomial
recursive sequences (Definition~\ref{deff:prs}) but on the right hand side of
the system of equations~\eqref{def:prs} we allow the $P_i$'s to be taken from
the field of fractions of the polynomial ring.  That is, each $P_i$ is of
the form $P_i(x_1,\ldots,x_k)=\frac{Q_i(x_1,\ldots,x_k)}{R_i(x_1,\ldots,x_k)}$,
where $Q_i,R_i\in \bbQ[x_1,\ldots,x_k]$ and $R_i\neq 0$.

This class extends both poly-recursive sequences and holonomic sequences (see Section~\ref{sec:introduction}).
For example one can express the sequence
of Catalan numbers, since $C_{n+1} = \frac{4n+2}{n+2}\cdot C_n$ and an ancillary
sequence can hold the value \(n\).  On the other hand,
the proof of the existence of cancelling polynomials for poly-recursive
sequences (Theorem~\ref{thm:cancelling}) carries over to rational recursive
sequences.  In particular, \(u_n = n^n\) is not even rational recursive.

This discussion points to the notion of rational
recursive sequences as a natural object for future research.

\subparagraph*{Acknowledgements.} We thank Maria Donten-Bury for suggesting the proof of Theorem~\ref{thm:cancelling} presented here. 
This proof replaced our previous more elaborate and less transparent argument. 
We also thank James Worrell, David Purser and Markus Whiteland for helpful
comments.
The research for this work was carried out in part at the
  Autob\'{o}z Research Camp in 2019 in Firbush, Scotland.
Finally, we thank the participants of the automata seminar at the University of Warsaw for an insightful discussion 
on the class of rational recursive sequences (considered in
Section~\ref{sec:conclusion}).

\bibliographystyle{plainurl}
\bibliography{bib2}

\begin{thebibliography}{10}

\bibitem{AlmagorBK11}
Shaull Almagor, Udi Boker, and Orna Kupferman.
\newblock What's decidable about weighted automata?
\newblock In {\em Automated Technology for Verification and Analysis, 9th
  International Symposium, {ATVA} 2011, Taipei, Taiwan, October 11-14, 2011.
  Proceedings}, pages 482--491, 2011.
\newblock URL: \url{https://doi.org/10.1007/978-3-642-24372-1\_37}, \href
  {http://dx.doi.org/10.1007/978-3-642-24372-1\_37}
  {\path{doi:10.1007/978-3-642-24372-1\_37}}.

\bibitem{klapki}
Ronald Alter and K.~K. Kubota.
\newblock Prime and prime power divisibility of {C}atalan numbers.
\newblock {\em Journal of Combinatorial Theory, Series A}, 15(3):243--256,
  1973.
\newblock URL:
  \url{http://www.sciencedirect.com/science/article/pii/0097316573900721},
  \href {http://dx.doi.org/https://doi.org/10.1016/0097-3165(73)90072-1}
  {\path{doi:https://doi.org/10.1016/0097-3165(73)90072-1}}.

\bibitem{AlurC10}
Rajeev Alur and Pavol Cern{\'{y}}.
\newblock Expressiveness of streaming string transducers.
\newblock In {\em {IARCS} Annual Conference on Foundations of Software
  Technology and Theoretical Computer Science, {FSTTCS} 2010, December 15-18,
  2010, Chennai, India}, pages 1--12, 2010.
\newblock URL: \url{https://doi.org/10.4230/LIPIcs.FSTTCS.2010.1}, \href
  {http://dx.doi.org/10.4230/LIPIcs.FSTTCS.2010.1}
  {\path{doi:10.4230/LIPIcs.FSTTCS.2010.1}}.

\bibitem{AlurDDRY13}
Rajeev Alur, Loris D'Antoni, Jyotirmoy~V. Deshmukh, Mukund Raghothaman, and
  Yifei Yuan.
\newblock Regular functions and cost register automata.
\newblock In {\em 28th Annual {ACM/IEEE} Symposium on Logic in Computer
  Science, {LICS} 2013, New Orleans, LA, USA, June 25-28, 2013}, pages 13--22,
  2013.
\newblock URL: \url{https://doi.org/10.1109/LICS.2013.65}, \href
  {http://dx.doi.org/10.1109/LICS.2013.65} {\path{doi:10.1109/LICS.2013.65}}.

\bibitem{baker1979trainable}
James~K Baker.
\newblock Trainable grammars for speech recognition.
\newblock {\em The Journal of the Acoustical Society of America},
  65(S1):S132--S132, 1979.

\bibitem{BarloyFLM20}
Corentin Barloy, Nathana{\"{e}}l Fijalkow, Nathan Lhote, and Filip Mazowiecki.
\newblock A robust class of linear recurrence sequences.
\newblock In {\em 28th {EACSL} Annual Conference on Computer Science Logic,
  {CSL} 2020, January 13-16, 2020, Barcelona, Spain}, pages 9:1--9:16, 2020.
\newblock URL: \url{https://doi.org/10.4230/LIPIcs.CSL.2020.9}, \href
  {http://dx.doi.org/10.4230/LIPIcs.CSL.2020.9}
  {\path{doi:10.4230/LIPIcs.CSL.2020.9}}.

\bibitem{BenediktDSW17}
Michael Benedikt, Timothy Duff, Aditya Sharad, and James Worrell.
\newblock Polynomial automata: Zeroness and applications.
\newblock In {\em 32nd Annual {ACM/IEEE} Symposium on Logic in Computer
  Science, {LICS} 2017, Reykjavik, Iceland, June 20-23, 2017}, pages 1--12,
  2017.
\newblock URL: \url{https://doi.org/10.1109/LICS.2017.8005101}, \href
  {http://dx.doi.org/10.1109/LICS.2017.8005101}
  {\path{doi:10.1109/LICS.2017.8005101}}.

\bibitem{BhattiproluGV17}
Vijay Bhattiprolu, Spencer Gordon, and Mahesh Viswanathan.
\newblock Extending parikh's theorem to weighted and probabilistic context-free
  grammars.
\newblock In {\em Quantitative Evaluation of Systems - 14th International
  Conference, {QEST} 2017, Berlin, Germany, September 5-7, 2017, Proceedings},
  pages 3--19, 2017.
\newblock URL: \url{https://doi.org/10.1007/978-3-319-66335-7\_1}, \href
  {http://dx.doi.org/10.1007/978-3-319-66335-7\_1}
  {\path{doi:10.1007/978-3-319-66335-7\_1}}.

\bibitem{DrosteG07}
Manfred Droste and Paul Gastin.
\newblock Weighted automata and weighted logics.
\newblock {\em Theor. Comput. Sci.}, 380(1-2):69--86, 2007.
\newblock URL: \url{https://doi.org/10.1016/j.tcs.2007.02.055}, \href
  {http://dx.doi.org/10.1016/j.tcs.2007.02.055}
  {\path{doi:10.1016/j.tcs.2007.02.055}}.

\bibitem{DrosteHWA09}
Manfred Droste, Werner Kuich, and Heiko Vogler.
\newblock {\em Handbook of Weighted Automata}.
\newblock Springer, 1st edition, 2009.

\bibitem{FerMarSen14}
Julien Fert{\'{e}}, Nathalie Marin, and G{\'{e}}raud S{\'{e}}nizergues.
\newblock Word-mappings of level 2.
\newblock {\em Theory Comput. Syst.}, 54(1):111--148, 2014.
\newblock URL: \url{https://doi.org/10.1007/s00224-013-9489-5}, \href
  {http://dx.doi.org/10.1007/s00224-013-9489-5}
  {\path{doi:10.1007/s00224-013-9489-5}}.

\bibitem{FraSen06}
S.~Fratani and G{\'{e}}raud S{\'{e}}nizergues.
\newblock Iterated pushdown automata and sequences of rational numbers.
\newblock {\em Ann. Pure Appl. Logic}, 141(3):363--411, 2006.
\newblock URL: \url{https://doi.org/10.1016/j.apal.2005.12.004}, \href
  {http://dx.doi.org/10.1016/j.apal.2005.12.004}
  {\path{doi:10.1016/j.apal.2005.12.004}}.

\bibitem{GantyG18}
Pierre Ganty and Elena Guti{\'{e}}rrez.
\newblock The {P}arikh property for weighted context-free grammars.
\newblock In {\em 38th {IARCS} Annual Conference on Foundations of Software
  Technology and Theoretical Computer Science, {FSTTCS} 2018, December 11-13,
  2018, Ahmedabad, India}, pages 32:1--32:20, 2018.
\newblock URL: \url{https://doi.org/10.4230/LIPIcs.FSTTCS.2018.32}, \href
  {http://dx.doi.org/10.4230/LIPIcs.FSTTCS.2018.32}
  {\path{doi:10.4230/LIPIcs.FSTTCS.2018.32}}.

\bibitem{Gerhold04}
Stefan Gerhold.
\newblock On some non-holonomic sequences.
\newblock {\em Electr. J. Comb.}, 11(1), 2004.
\newblock URL:
  \url{http://www.combinatorics.org/Volume\_11/Abstracts/v11i1r87.html}.

\bibitem{halava2005skolem}
Vesa Halava, Tero Harju, Mika Hirvensalo, and Juhani Karhum{\"a}ki.
\newblock Skolem’s problem-on the border between decidability and
  undecidability.
\newblock Technical report, Technical Report 683, Turku Centre for Computer
  Science, 2005.

\bibitem{KauersP11}
Manuel Kauers and Peter Paule.
\newblock {\em The Concrete Tetrahedron - Symbolic Sums, Recurrence Equations,
  Generating Functions, Asymptotic Estimates}.
\newblock Texts {\&} Monographs in Symbolic Computation. Springer, 2011.
\newblock URL: \url{https://doi.org/10.1007/978-3-7091-0445-3}, \href
  {http://dx.doi.org/10.1007/978-3-7091-0445-3}
  {\path{doi:10.1007/978-3-7091-0445-3}}.

\bibitem{KreutzerR13}
Stephan Kreutzer and Cristian Riveros.
\newblock Quantitative monadic second-order logic.
\newblock In {\em 28th Annual {ACM/IEEE} Symposium on Logic in Computer
  Science, {LICS} 2013, New Orleans, LA, USA, June 25-28, 2013}, pages
  113--122, 2013.
\newblock URL: \url{https://doi.org/10.1109/LICS.2013.16}, \href
  {http://dx.doi.org/10.1109/LICS.2013.16} {\path{doi:10.1109/LICS.2013.16}}.

\bibitem{Lang}
Serge Lang.
\newblock {\em Algebra}.
\newblock Graduate Texts in Mathematics. Springer, 2002.

\bibitem{MazowieckiR19}
Filip Mazowiecki and Cristian Riveros.
\newblock Copyless cost-register automata: Structure, expressiveness, and
  closure properties.
\newblock {\em J. Comput. Syst. Sci.}, 100:1--29, 2019.
\newblock URL: \url{https://doi.org/10.1016/j.jcss.2018.07.002}, \href
  {http://dx.doi.org/10.1016/j.jcss.2018.07.002}
  {\path{doi:10.1016/j.jcss.2018.07.002}}.

\bibitem{OuaknineW15}
Jo{\"{e}}l Ouaknine and James Worrell.
\newblock On linear recurrence sequences and loop termination.
\newblock {\em {SIGLOG} News}, 2(2):4--13, 2015.
\newblock URL: \url{https://dl.acm.org/citation.cfm?id=2766191}.

\bibitem{Senizergues07}
G{\'{e}}raud S{\'{e}}nizergues.
\newblock Sequences of level 1, 2, 3, ..., \emph{k} , ..
\newblock In {\em Computer Science - Theory and Applications, Second
  International Symposium on Computer Science in Russia, {CSR} 2007,
  Ekaterinburg, Russia, September 3-7, 2007, Proceedings}, pages 24--32, 2007.
\newblock URL: \url{https://doi.org/10.1007/978-3-540-74510-5\_6}, \href
  {http://dx.doi.org/10.1007/978-3-540-74510-5\_6}
  {\path{doi:10.1007/978-3-540-74510-5\_6}}.

\end{thebibliography}
\end{document}